\documentclass[11pt,twoside,letterpaper]{article}
\usepackage{hyperref}
\usepackage{graphicx}
\usepackage{amsmath}
\usepackage{amsthm}
\usepackage{amssymb}
\numberwithin{equation}{section}
\usepackage{color}
\theoremstyle{plain}
\newtheorem{proposition}{Proposition}

\newtheorem{lemma}[proposition]{Lemma}

\theoremstyle{definition}

\newtheorem{example}[proposition]{Example}
\newtheorem{remark}[proposition]{Remark}

  

\makeatletter
\def\cleardoublepage{\clearpage\if@twoside \ifodd\c@page\else%
    \hbox{}%
    \thispagestyle{empty}%
    \newpage%
    \if@twocolumn\hbox{}\newpage\fi\fi\fi}
\makeatother

\def\figurename{Figure}
\makeatletter
\renewcommand{\fnum@figure}[1]{\figurename~\thefigure.}
\makeatother

\def\tablename{Table}
\makeatletter
\renewcommand{\fnum@table}[1]{\tablename~\thetable.}
\makeatother

\begin{document}

\title{
{ \vskip 0.45in \bfseries\scshape On the peakon inverse problem for the Degasperis-Procesi equation}}
\author{\bfseries\itshape Keivan Mohajer \thanks{Department of Mathematics, University of Isfahan, Isfahan, 81746-73441, Iran; k.mohajer@sci.ui.ac.ir}
%\and \bfseries\itshape Jacek Szmigielski\thanks{Department of
%Mathematics and Statistics, University of Saskatchewan, 106
%Wiggins Road, Saskatoon, Saskatchewan, S7N 5E6, Canada;
%szmigiel@math.usask.ca}
 }

\date{}

\maketitle
\begin{abstract}
  The peakon inverse problem for the Degasperis-Procesi equation
  is solved directly on the real line, using Cauchy biorthogonal polynomials, without any additional transformation to a ``string" type boundary value problem known from prior works.
  \end{abstract}

\section{Introduction}
The Degasperis-Procesi (DP) equation \cite{degasperis-procesi}
\begin{equation}\label{DP-equation}
u_t-u_{xxt}+4uu_x=3u_xu_{xx}+uu_{xxx},\quad (x,t)\in \mathbb{R}^2,
\end{equation}
like the other two similar nonlinear equations, namely, the
Camassa-Holm (CH) \cite{camassa-holm} and Novikov's equation
\cite{novikov}, admits a type of nonsmooth solution called
$n$-peakon. $n$-peakon solutions of nonlinear equations have been
the subject of research (see for example
\cite{beals-sattinger-szmigielski-moment},
\cite{beals-sattinger-szmigielski-stieltjes},
\cite{lundmark-szmigielski-DPlong} and \cite{hls}) mostly because
of their interesting soliton-like behavior (for the behavior of
soliton solutions of nonlinear equations see
\cite{constantin-escher}, \cite{constantin-ivanov-lenells},
\cite{constantin-mckean}, \cite{constantin-strauss} and
\cite{constantin2000}). In particular, the wave breaking phenomena
was studied in \cite{constantin-escher} and \cite{constantin2000},
and stability of peakons was studied in \cite{constantin-strauss}.
Also, peakons can be viewed as limit of solitary waves (see for
example \cite{yu-tian-wang}). For different views on the solutions
of the DP equation one can refer to \cite{qiao} and
\cite{vakhnenko-parkes}. An $n$-peakon is a solution of the
following form
\begin{equation}\label{n-peakon}
u(x,t)=\sum_{n=1}^{n} m_i(t) e^{-|x-x_i(t)|}.
\end{equation}
For the DP equation it is known that $x_i(t)$ and $m_i(t)$ must satisfy the following system of nonlinear ODEs:
\begin{align}
& \dot{x}_j = \sum_{i=1}^{n} m_i e^{-|x_j-x_i|},\quad j=1,\dots, n,\label{xmODE1}\\
& \dot{m}_j = 2\sum_{i=1}^{n}m_j m_i \operatorname{sgn}(x_j-x_i)
e^{-|x_j-x_i|},\quad j=1,\dots, n.\label{xmODE2}
\end{align}
The peakon inverse problem method provides the solution of
\eqref{xmODE1}-\eqref{xmODE2} and consequently it provides the
$n$-peakon solutions of the DP equation. Previously, for CH, DP
and the Novikov's equation the inverse problem was solved using a
transformation which translates the problem on the real line into
a problem on a finite interval. (see
\cite{beals-sattinger-szmigielski-stieltjes},
\cite{lundmark-szmigielski-DPlong} and \cite{hls} respectively.)
However, recently the peakon inverse problem for the two
integrable equations CH and Novikov's equation was solved in
\cite{Mohajer-Szmigielski-inverse-CH} and
\cite{Mohajer-Szmigielski-inverse-Novikov}, respectively, without
any transformation of the problem to a string type boundary value
problem. In particular, in
\cite{Mohajer-Szmigielski-inverse-Novikov} it was shown that for
the Novikov's equation, Cauchy biorthogonal polynomials
\cite{bgs1} are the solutions to the approximation problem
relevant to the inverse problem.
Recent developments (see
\cite{bgs3}, \cite{bgs2}, \cite{MR3088819}, \cite{MR3162486} and
\cite{bertola-bothner}) suggest that Cauchy biorthogonal
polynomials can be also useful in the investigation of various
problems in random matrix theory.  In this paper it is
shown that Cauchy biorthogonal polynomials \cite{bgs1} can be used
to solve the peakon inverse problem arising in the DP equation. The associated boundary value
problem is non-selfadjoint, given by a third order differential equation.  So it is indeed very interesting and important
to verify that these polynomials can be applied to solve the inverse problem of the non-selfadjoint type.
\section{Forward problem}
It is easy to verify that another form of the DP equation \eqref{DP-equation} is
\begin{align}\label{DP-m}
&m_t + m_x u +3 mu_x=0,\\
&m=u-u_{xx}.
\end{align}
So, if the $n$-peakon solution \eqref{n-peakon} satisfies the DP equation, we must have $m=2\sum_{i=1}^n m_i \delta_{x_i}$.
It is known that (See \cite{degasperis-holm-hone}) the DP equation is the compatibility condition for the following system for $\psi(x,t;z)$:
\begin{align}
& (\partial_x-\partial_x^3)\psi=zm\psi,\label{Lax-pair-1}\\
& \psi_t=\big(\frac{1}{z}(1-\partial_x^2)+u_x-u\partial_x\big)\psi.\label{Lax-pair-2}
\end{align}
If $m=2\sum_{i=1}^n m_i \delta_{x_i}$, then \eqref{Lax-pair-1} implies that in every open interval $(x_i,x_{i+1})$ we have $(\partial_x-\partial_x^3)\psi=0$. Therefore, on such an interval we have
\begin{equation}\label{local-psi}
\psi(x,t;z)=A_i(t;z)e^x + B_i(t;z) + C_i(t;z)e^{-x},\quad x\in (x_i,x_{i+1}),\quad i=0,1,\dots,n,
\end{equation}
where $x_0=-\infty$ and $x_{n+1}=+\infty$. Equation
$(\partial_x-\partial_x^3)\psi=2z\sum_{i=1}^n m_i
\psi(x_i,t;z)\delta_{x_i}$ also implies that (explanation can be
found in \cite{lundmark-szmigielski-DPlong})
\begin{equation}\label{ABC-coeff}
\begin{pmatrix}
A_k(t;z) \\ B_k(t;z) \\ C_k(t;z)
\end{pmatrix} = S_k(z)S_{k-1}(z)\dots S_2(z)S_1(z)
\begin{pmatrix}
1 \\ 0 \\ 0
\end{pmatrix},\ \ \ \ \ \ \ k=1,\dots,n.
\end{equation}
where
\begin{equation}\label{S-equation}
S_k(z)=I-zm_k
\begin{pmatrix}
e^{-x_k} \\ -2 \\ e^{x_k}
\end{pmatrix}
\begin{pmatrix}
e^{x_k} & 1 & e^{-x_k}
\end{pmatrix}.
\end{equation}
It can be verified that $S_k(z)^{-1}=S_k(-z)$, $\det(S_k(z))=1$ and $L S_k^T(-z) L^{-1}=S_k(-z)$ where
\begin{equation*}
L=\begin{pmatrix}
0 & 0 & 1\\
0 & -2& 0\\
1 & 0 & 0
\end{pmatrix}.
\end{equation*}
Now if we set
$S_{[n,k]}=S_n \dots S_{n-k+1}$, then
\begin{equation}\label{recursion}
\begin{pmatrix}
A_n \\ B_n \\ C_n
\end{pmatrix} =S_{[n,k]}
\begin{pmatrix}
A_{n-k} \\ B_{n-k} \\ C_{n-k}
\end{pmatrix}.
\end{equation}
In order to proceed we need the following lemma.
\begin{lemma}\label{L-recursion}
\begin{equation}
S_{[n,k]}^{-1}(z)=LS_{[n,k]}^{T}(-z)L^{-1}
\end{equation}
\end{lemma}
\begin{proof}
We can write
\begin{align*}
L S_{[n,k]}^T(-z) L^{-1} &=L S_{n-k+1}^T(-z) \dots S_{n}^T(-z)\\
&=\bigl(L S_{n-k+1}^T(-z) L^{-1}\bigr)\bigl(L S_{n-k+2}^T(-z) L^{-1}\bigr)\dots \bigl(L S_{n}^T(-z) L^{-1}\bigr)\\
&=S_{n-k+1}(-z) \dots S_{n}(-z)\\
&=S_{n-k+1}^{-1}(z) \dots S_{n}^{-1}(z)\\
&=S_{[n,k]}^{-1}(z).
\end{align*}
\end{proof}
Thus, one can show that all the entries of the
matrix $S_{[n,k]}$ and its adjoint are of degree $k$. We denote
the entries of $S_{[n,k]}$ by $s_{ij}$. In the following propositions we obtain approximations that will be needed for the inverse problem.
\begin{proposition}\label{prop1}
\begin{equation}\label{approx-A-B}
\frac{B_n}{A_n}-\frac{s_{21}}{s_{11}}=\mathcal{O}\biggl(\frac{1}{z^k}\biggr),\
\ \ \ \ z\to \infty,
\end{equation}
\begin{equation}
\frac{C_n}{A_n}-\frac{s_{31}}{s_{11}}=\mathcal{O}\biggl(\frac{1}{z^k}\biggr),\
\ \ \ \ z\to \infty.
\end{equation}
\end{proposition}
\begin{proof}
From \eqref{recursion} we have
\begin{equation*}
\begin{split}
\frac{B_n}{A_n}&=\frac{s_{21}A_{n-k}+s_{22}B_{n-k}+s_{23}C_{n-k}}{s_{11}A_{n-k}+s_{12}B_{n-k}+s_{13}C_{n-k}}\\
&=\frac{s_{21}+s_{22}\frac{B_{n-k}}{A_{n-k}}+s_{23}\frac{C_{n-k}}{A_{n-k}}}{s_{11}+s_{12}\frac{B_{n-k}}{A_{n-k}}+s_{13}\frac{C_{n-k}}{A_{n-k}}}.
\end{split}
\end{equation*}
Therefore,
\begin{equation*}
\frac{B_n}{A_n}-\frac{s_{21}}{s_{11}}=\frac{(s_{11}s_{22}-s_{12}s_{21})\frac{B_{n-k}}{A_{n-k}}+(s_{11}s_{23}-s_{13}s_{21})\frac{C_{n-k}}{A_{n-k}}}{s_{11}(s_{11}+s_{12}\frac{B_{n-k}}{A_{n-k}}+s_{13}\frac{C_{n-k}}{A_{n-k}})}.
\end{equation*}
Now, since all the entries of $S_{[n,k]}$ and its adjoint are of degree $k$ and
\begin{align*}
&\frac{B_{n-k}}{A_{n-k}}=\mathcal{O}(1),\quad z\to \infty,\\
&\frac{C_{n-k}}{A_{n-k}}=\mathcal{O}(1),\quad z\to \infty,
\end{align*}
We get
\begin{equation*}
\frac{B_n}{A_n}-\frac{s_{21}}{s_{11}}=\mathcal{O}\biggl(\frac{1}{z^k}\biggr),\quad z\to \infty.
\end{equation*}
A similar argument proves the second approximation.
\end{proof}
Now we consider the Weyl functions of the DP equation that are
introduced in \cite{lundmark-szmigielski-DPlong}, namely
\begin{equation}\label{Weyl-functions-def}
W(z)=-\frac{B(z)}{2zA(z)},\ \ \ \ \ \ \ \
Z(z)=\frac{C(z)}{2zA(z)},
\end{equation}
where $A(z)=A_n(z)$, $B(z)=B_n(z)$ and $C(z)=C_n(z)$. $A_n$, $B_n$
and $C_n$ are given by equation \eqref{ABC-coeff}. Note that in \cite{lundmark-szmigielski-DPlong} these two Weyl functions are denoted by $\omega$ and $\zeta$ respectively.
From proposition \ref{prop1} and  equations \eqref{Weyl-functions-def} we get the approximations
\begin{equation}\label{W-approx-s}
W(z)+\frac{s_{21}}{2zs_{11}}=\mathcal{O}\biggl(\frac{1}{z^{k+1}}\biggr),\
\ \ \ \ z\to \infty,
\end{equation}
\begin{equation}
Z(z)-\frac{s_{31}}{2zs_{11}}=\mathcal{O}\biggl(\frac{1}{z^{k+1}}\biggr),\
\ \ \ \ z\to \infty.
\end{equation}
We note that $s_{11}(0)=1$, $s_{21}(0)=0$ and $s_{31}(0)=0$.

\begin{proposition}
\begin{equation}\label{ABC-WZ}
\frac{C_n(z)}{A_n(z)}s_{11}(-z)-\frac{1}{2}\frac{B_n(z)}{A_n(z)}s_{21}(-z)+s_{31}(-z)=\mathcal{O}\biggl(\frac{1}{z^{k}}\biggr),\
\ \ \ \ z\to \infty.
\end{equation}
\end{proposition}
\begin{proof}
By lemma \ref{L-recursion} we have
\begin{equation*}
L S_{[n,k]}^T(-z) L^{-1}
\begin{pmatrix}
A_n \\ B_n \\ C_n
\end{pmatrix}=
\begin{pmatrix}
A_{n-k} \\ B_{n-k} \\ C_{n-k}
\end{pmatrix}.
\end{equation*}
So,
\begin{equation*}
S_{[n,k]}^T(-z) L^{-1}
\begin{pmatrix}
A_n \\ B_n \\ C_n
\end{pmatrix}=L^{-1}
\begin{pmatrix}
A_{n-k} \\ B_{n-k} \\ C_{n-k}
\end{pmatrix}.
\end{equation*}
Taking the transpose of both sides we get
\begin{equation*}
\begin{pmatrix}
A_n & B_n & C_n
\end{pmatrix}L^{-1}S_{[n,k]}(-z)=\begin{pmatrix}
A_{n-k} & B_{n-k} & C_{n-k}
\end{pmatrix}L^{-1},
\end{equation*}
or
\begin{equation*}
\begin{pmatrix}
C_n & -\frac{1}{2}B_n & A_n
\end{pmatrix}S_{[n,k]}(-z)=\begin{pmatrix}
C_{n-k} & -\frac{1}{2}B_{n-k} & A_{n-k}
\end{pmatrix}.
\end{equation*}
Hence,
\begin{equation*}
C_n(z)s_{11}(-z)-\frac{1}{2}B_n(z)s_{21}(-z)+A_n(z)s_{31}(-z)=C_{n-k}(z).
\end{equation*}
So,
\begin{equation*}
\frac{C_n(z)}{A_n(z)}s_{11}(-z)-\frac{1}{2}\frac{B_n(z)}{A_n(z)}s_{21}(-z)+s_{31}(-z)=\frac{C_{n-k}(z)}{A_n(z)}=\mathcal{O}\biggl(\frac{1}{z^k}\biggr),
\quad z\to \infty.
\end{equation*}
\end{proof}

It is clear that $A_n(z)$, $B_n(z)$ and $C_n(z)$ are
polynomials in $z$ of degree $n$. Therefore, using the partial
fraction decomposition we can write
\begin{equation}\label{partial-decomp-general-W}
W(z)=-\frac{B(z)}{2zA(z)}=\sum_{k=1}^{p}\sum_{i=1}^{n_k}\frac{b_{ki}}{(z-\lambda_k)^i},
\end{equation}
\begin{equation}\label{partial-decomp-general-Z}
Z(z)=\frac{C(z)}{2zA(z)}=\sum_{k=1}^{p}\sum_{i=1}^{n_k}\frac{c_{ki}}{(z-\lambda_k)^i},
\end{equation}
where each $\lambda_k$ is a root of multiplicity $n_k$ for
$A_n(z)$ and $\sum_{k=1}^{p}n_k=n$. Now, from theorem $2.12$ in
\cite{lundmark-szmigielski-DPlong} we understand that with the
assumption of all $m_i(0)>0$ and $x_1(0)<\dots<x_n(0)$, we obtain real
distinct positive roots for $A_n(z)$ i.e.
$0<\lambda_1<\dots<\lambda_n$. As in
\cite{lundmark-szmigielski-DPlong} we call the above assumption, ``the pure peakon
assumption (PPA)''. For the time being we suppose the pure peakon
assumption holds unless otherwise stated. Therefore we can
introduce the discrete measure
$\mu(x)=\sum_{k=1}^{n}b_{k}\delta_{\lambda_k}(x)$. Then we have
\begin{equation}\label{partial-decomp-W}
W(z)=\sum_{k=1}^n\frac{b_k}{z-\lambda_k}= \int \frac{1}{z-x}\ d\mu(x),
\end{equation}
and
\begin{equation}\label{partial-decomp-Z}
Z(z)=\sum_{k=1}^{n}\frac{c_{k}}{z-\lambda_k}.
\end{equation}
According to a proposition in \cite{lundmark-szmigielski-DPlong}, if $x_i(t)$ and $m_i(t)$ satisfy the system of ODEs \eqref{xmODE1}-\eqref{xmODE2}, then
\begin{equation}\label{time-evolution-ABC}
\dot{A}=0,\quad \dot{B}=\frac{A}{z}-2AM_+,\quad \dot{C}=-BM_+.
\end{equation}
Applying equations \eqref{time-evolution-ABC}
to equations \eqref{partial-decomp-W} and
\eqref{partial-decomp-Z} one
can prove the following formulas (a detailed discussion can be
found in \cite{lundmark-szmigielski-DPlong})
%\begin{equation}
%\begin{split}
%&\dot{b}_{kj}=\sum_{i=j}^{n_k}\frac{(-1)^{i-j}b_{ki}}{\lambda_k^{i-j+1}},\\
%&M_+=\sum_{k=1}^{p} \dot{b}_{k1},\\
%&\dot{c}_{lj}=b_{lj}M_+=b_{lj}\sum_{k=1}^{p} \dot{b}_{k1},
%\end{split}
%\end{equation}

\begin{equation}
\begin{split}
&\dot{b}_{k}=\frac{b_{k}}{\lambda_k},\\
&M_+=\sum_{k=1}^{n} \dot{b}_{k},\\
&\dot{c}_{l}=b_{l}M_+=b_{l}\sum_{k=1}^{n} \dot{b}_{k},
\end{split}
\end{equation}
where $M_+ = \sum_{i=1}^{n}m_i(t)e^{x_i(t)}$.
Therefore one can proceed as in \cite{lundmark-szmigielski-DPlong}
to get the time evolution of $b_k$ and $c_k$ as follows:
\begin{equation}
\begin{split}\label{bc-time-evolution}
&b_k(t)=b_k(0)e^{t_k/\lambda_k},\\
&c_k(t)=\sum_{i=1}^{n}\frac{\lambda_k b_k(t) b_i(t)}{\lambda_i +
\lambda_k}.
\end{split}
\end{equation}
Hence, using equation \eqref{partial-decomp-Z} and the second equation of \eqref{bc-time-evolution} we can write
\begin{equation}
Z(z)=\iint \frac{x}{(z-x)(x+y)}\ d\mu(x)d\mu(y).
\end{equation}
\section{Inverse problem}
The peakon inverse problem can be formulated as follows:\\
Suppose the discrete measure
$\mu(x)=\sum_{k=1}^{n}b_{k}\delta_{\lambda_k}(x)$ is given and let
\begin{equation}
W(z)=\int \frac{1}{z-x}\ d\mu(x),\quad Z(z)=\iint
\frac{x}{(z-x)(x+y)}\ d\mu(x)d\mu(y).
\end{equation}
Find polynomials $P_k(z)$, $Q_k(z)$ and $\hat{P}_k(z)$ such that
\begin{equation}\label{WPQ}
W(z)-\frac{P_k(z)}{Q_k(z)}=\mathcal{O}(\frac{1}{z^{k+1}}),\ \ \ \
\ z\to \infty,
\end{equation}
\begin{equation}
Z(z)-\frac{\hat{P}_k(z)}{Q_k(z)}=\mathcal{O}(\frac{1}{z^{k+1}}),\
\ \ \ \ z\to \infty,
\end{equation}
\begin{equation}\label{ZQWP}
Z(-z)Q_k(z)-zW(-z)P_k(z)-\hat{P}_k(z)=\mathcal{O}(\frac{1}{z^k}),\
\ \ \ \ z\to \infty,
\end{equation}
with $\deg(P_k)=\deg(\hat{P}_k)=k$, $\deg(Q_k)=k+1$, $P_k(0)=0$,
$\hat{P}_k(0)=0$, $Q_k(0)=0$ and $Q_k^\prime(0)=2$.
Now consider the polynomials $p_k(z)=P_k(z)/z$, $\hat{p}_k(z)=\hat{P}_k(z)/z$ and $q_k(z)=Q_k(z)/z$. Then the approximation problem can be rewritten as follows
\begin{align}
& W(z)-\frac{p_k(z)}{q_k(z)}=\mathcal{O}(\frac{1}{z^{k+1}}),\quad z\to \infty,\label{approx-p-q}\\
& Z(z)-\frac{\hat{p}_k(z)}{q_k(z)}=\mathcal{O}(\frac{1}{z^{k+1}}),\quad z\to \infty,\label{approx-phat-q}\\
& Z(-z)q_k(z)-z W(-z)p_k(z)-\hat{p}_k(z)=\mathcal{O}(\frac{1}{z^{k+1}}),\quad z\to \infty,\label{approx-phat-p-q}
\end{align}
with $\deg(p_k)=\deg(\hat{p}_k)=k-1$, $\deg(q_k)=k$, and
$q_k(0)=2$. The next step is to transform the approximation problem \eqref{approx-p-q}--\eqref{approx-phat-p-q} into another approximation problem in order to apply theorem 5.1 in in \cite{bgs1}.
First, we need the following notation
\begin{equation}
\gamma_j=\int x^j d\mu (x),\quad I_{i,j}=\iint
\frac{x^{i+1}y^j}{x+y}d\mu (x) d\mu (y).
\end{equation}
Now the approximation problem
\eqref{approx-p-q}--\eqref{approx-phat-p-q} can be easily transformed into the following approximation problem.

\begin{align}
&W(z)-\frac{p_k(z)}{q_k(z)}=\mathcal{O}(\frac{1}{z^{k+1}}),\quad z\to \infty,\label{approx-e1}\\
&\gamma_0 W(z) - Z(z) - \frac{\gamma_0 p_k(z)-\hat{p}_k(z)}{q_k(z)}=\mathcal{O}(\frac{1}{z^{k+1}}),\quad z\to \infty,\label{approx-e2}\\
&Z(-z)q_k(z) - \left(\int\frac{x}{z+x}\,d\mu(x) \right) p_k(z) +
\gamma_0
p_k(z)-\hat{p}_k(z)=\mathcal{O}(\frac{1}{z^{k+1}}),\label{approx-e3}
\end{align}
where $\gamma_0=\int \,d\mu(x) = \sum_{k=1}^{n} b_k$.

At this point, one can easily observe that by setting $d\alpha(x)=x d\mu(x)$ and $d\beta(x)=d\mu(x)$, in definition 5.1 in \cite{bgs1}, we obtain the following
\begin{equation}\label{Nikishin}
\begin{split}
W_\beta (z)&=W(z)=\int\frac{1}{z-x}\,d\mu(x),\\
W_{\beta \alpha^*}(z)&=\gamma_0 W(z) - Z(z) =\iint \frac{y}{(z-x)(x+y)}\,d\mu(x)\,d\mu(y)\\ &=\iint \frac{x}{(z-y)(x+y)}\,d\mu(x)\,d\mu(y),\\
W_{\alpha^* \beta}(z)&=Z(-z)=\iint\frac{-x}{(z+x)(x+y)}\,d\mu(x)\,d\mu(y),\\
W_{\alpha^*}(z)&=\int \frac{x}{z+x}\,d\mu(x).
\end{split}
\end{equation}
Then the approximation problem \eqref{approx-e1}--\eqref{approx-e3} follows from definition 5.3 in \cite{bgs1}. So according to theorem 5.1 in \cite{bgs1} a sequence of Cauchy biorthogonal polynomials is the unique, up to a multiplicative constant, solution of the approximation problem \eqref{approx-e1}--\eqref{approx-e3}.
Therefore, using theorem 5.1 in \cite{bgs1} and the normalization
condition $q_k(0)=2$ we have
\begin{equation}\label{qk-det-formula}
q_k(z)=\frac{2}{\Delta_k^{01}}
\begin{vmatrix}
1 &z &\dots & z^k\\
I_{0,0} & I_{0,1} & \dots & I_{0,k}\\
\vdots & \vdots & \vdots & \vdots\\
I_{k-1,0} & I_{k-1,1} & \dots & I_{k-1,k}
\end{vmatrix},
\end{equation}
where
\begin{equation*}
\Delta_k^{ab}=
\begin{vmatrix}
I_{a,b} & I_{a,b+1} & \dots & I_{a,b+k-1}\\
I_{a+1,b} & I_{a+1,b+1} & \dots & I_{a+1,b+k-1}\\
\vdots & \vdots & \vdots & \vdots\\
I_{a+k-1,b} & I_{a+k-1,b+1} & \dots & I_{a+k-1,b+k-1}
\end{vmatrix}.
\end{equation*}
Also, we have
\begin{align}
p_k(z)&=\int \frac{q_k(z)-q_k(y)}{z-y}d\mu(y),\label{pk-formula}\\
\gamma_0 p_k(z)-\hat{p}_k(z)&=\iint  \frac{(q_k(z)-q_k(y))y}{(z-y)(x+y)}d\mu(x)
d\mu(y).
\end{align}
Now we summarize the steps to find the momenta $m_i(t)$ and the
locations $x_i(t)$ for the n-peakon solution \eqref{n-peakon} of
the DP equation. First, we need a notation. If $p(z)$ is a
polynomial, then by $p[i]$ we mean the coefficient of the $i$th
power of $z$ in $p(z)$. Equations \eqref{ABC-coeff},
\eqref{S-equation} and \eqref{recursion} imply that for the
entries of $S_{[n,k]}$ we have $s_{11}[1]=-\sum_{i=n-k+1}^n m_i$
and $s_{12}[1]=2\sum_{i=n-k+1}^n m_i e^{x_i}$. Now comparing the
approximations \eqref{WPQ} and \eqref{ZQWP} with
\eqref{W-approx-s} and \eqref{ABC-WZ} we see that $2s_{11}=q_k(z)$
and $s_{12}=-zp_k(z)$. Therefore, from equation
\eqref{qk-det-formula} we get
\begin{equation}\label{formula-M}
\sum_{i=n-k+1}^n m_i =-\frac{q_k[1]}{2} =\frac{1}{\Delta_k^{01}}
\begin{vmatrix}
I_{0,0} & I_{0,2} &I_{0,3} & \dots & I_{0,k}\\
I_{1,0} & I_{1,2} &I_{1,3} &\dots & I_{1,k}\\
\vdots & \vdots & \vdots &\vdots &\vdots\\
I_{k-1,0} & I_{k-1,2} &I_{k-1,3} &\dots & I_{k-1,k}
\end{vmatrix}.
\end{equation}
Also, equations \eqref{pk-formula} and \eqref{qk-det-formula} imply that
\begin{equation}\label{formula-M+}
\sum_{i=n-k+1}^n m_i e^{x_i} =-\frac{p_k[0]}{2}=-\frac{1}{\Delta_k^{01}}
\begin{vmatrix}
0 &\gamma_0 &\dots & \gamma_{k-1}\\
I_{0,0} & I_{0,1} & \dots & I_{0,k}\\
\vdots & \vdots & \vdots & \vdots\\
I_{k-1,0} & I_{k-1,1} & \dots & I_{k-1,k}
\end{vmatrix}.
\end{equation}
Note that when $k=1$, we get
\begin{equation}
m_n=\frac{I_{0,0}}{I_{0,1}},\quad m_n e^{x_n} = \frac{\gamma_{0}I_{0,0}}{I_{0,1}}.
\end{equation}
So, $x_n=\ln (\gamma_{0})$.
Hence, starting from $k=1$, successive application of formulas \eqref{formula-M} and \eqref{formula-M+} will recover $m_i(t)$ and $x_i(t)$ for $i=1,\dots n$.
\begin{example}
For $n=2$ the $2$-peakon solution can be obtained as follows
\begin{equation}
\begin{split}
& x_1 = \ln \frac{b_1 b_2(\lambda_2-\lambda_1)^2}{b_1 \lambda_1^2
+ b_2\lambda_2^2 +\lambda_1\lambda_2(b_1+b_2)},\\
&m_1 = \frac{(\lambda_1 + \lambda_2)(b_1\lambda_1 + b_2
\lambda_2)^2}{\lambda_1\lambda_2[(b_1\lambda_1 + b_2 \lambda_2)^2
+ \lambda_1\lambda_2(b_1 + b_2)^2]},\\
&x_2 = \ln(b_1 + b_2),\\
&m_2 = \frac{(b_1 + b_2)^2(\lambda_1 + \lambda_2)}{(b_1\lambda_1 +
b_2 \lambda_2)^2 + \lambda_1\lambda_2(b_1 + b_2)^2}.
\end{split}
\end{equation}
This is the same as the $2$-peakon solution $(2.7)$ in
\cite{lundmark-szmigielski-DPlong}.
\end{example}
\begin{remark}
Also it can be verified that (H. Lundmark, personal communication,
June 3, 2014) for every $n$ the solution of the inverse problem is
identical to the solution in \cite{lundmark-szmigielski-DPlong}.
\end{remark}

\section{Conclusions}
In this paper we have solved the peakon inverse problem for the DP equation, under the pure peakon assumption, directly on the real line using Cauchy biorthogonal polynomials. An open problem related to this work is to formulate and solve an inverse problem for the shock peakon solutions \cite{lundmark-shockpeakons} of the DP equation.

\section{Acknowledgments}
I would like to thank Professor H. Lundmark for tremendously
helpful suggestions and comments.

\bibliographystyle{abbrv}
\bibliography{IP}
\end{document}